\documentclass[a4paper,UKenglish]{lipics-v2016}

\usepackage{microtype}
\usepackage{float}
\usepackage{verbatim}

\usepackage{amsfonts,amsmath,amssymb}
\usepackage{xcolor,hyperref}
\usepackage{enumerate,textcomp}
\usepackage{xspace}
\usepackage{pgf,tikz}
\usetikzlibrary{arrows,automata,shapes}
\tikzstyle{every state}=[minimum size=12pt,inner sep=0pt]

\newcommand{\N}{{\mathbb N}}

\newcommand{\dz}{\mathfrak d}
\newcommand{\aut}[1]{{\mathcal #1}}
\newcommand{\auta}{\mathcal{A}}
\newcommand{\dual}[1]{{\mathfrak d}({#1})}
\newcommand{\mot}[1]{{\mathbf {#1}}}

\newcommand{\pres}[1]{\langle{#1}\rangle}
\newcommand{\presm}[1]{\pres{{#1}}_{+}}

\newcommand{\otree}[1][]{\mathfrak{t}{\ifthenelse{\equal{#1}{}}{}{(\aut{#1})}}}

\newcommand{\ota}{\otree{(\auta)}}

\newcommand{\gauta}{\pres{\auta}}
\newcommand{\ie}{\emph{i.e.}\xspace}

\newcommand{\Sufeq}{\mathcal{S}_{\text{eq}}}
\newcommand{\Preeq}{\mathcal{P}_{\text{eq}}}

\newcommand{\jungle}{{\mathfrak j}}

\newcommand{\Li}[1]{\mot{#1 }L_{\mot{#1}}}	
	
\newcommand{\equivclass}{\gamma}

\newcommand{\lacroix}{\tikz[baseline=-.5ex]{\draw[->,>=latex] (0,0) -- (4ex,0); \draw[->,>=latex] (1.8ex,2ex) -- (1.8ex,-2ex);}}

\usepackage{xcolor}

\theoremstyle{plain}
\newtheorem{proposition}[theorem]{Proposition}
\newtheorem{myremark}[theorem]{Remark}

\newtheorem*{theorem*}{Theorem}

\title{Connected reversible Mealy automata of prime size cannot generate infinite Burnside groups
\footnote{This work was partially supported by the French \emph{Agence Nationale pour la~Recherche},
through the Project {\bf MealyM} ANR-JS02-012-01.
}}
\titlerunning{Connected bir. automata of prime size cannot generate infinite Burnside groups}
\author{Thibault Godin}
\author{Ines Klimann}
\affil{Univ. Paris Diderot, Sorbonne Paris Cit\'e, IRIF, UMR 8243 CNRS,\\ F-75013 Paris, France
\texttt{\{godin,klimann\}@liafa.univ-paris-diderot.fr}}%
\authorrunning{Th. Godin and I. Klimann}

\Copyright{{Th. Godin and I. Klimann}}


\EventEditors{}

\EventLongTitle{}
\EventShortTitle{}
\EventAcronym{}
\EventYear{}
\EventDate{}
\EventLocation{}
\EventLogo{}
\ArticleNo{1}

\usepackage{amsmath}

\begin{document}

\maketitle

\begin{abstract}
The simplest example of an infinite Burnside group arises in the class of automaton groups. 
However there is no known example of such a group generated by a
reversible Mealy automaton. It has been proved that, for a connected automaton of size at most~3, or when the automaton is not bireversible, the generated group cannot be Burnside infinite. In this paper, we extend these results to  automata with bigger stateset, proving that, if a connected reversible automaton  has a prime number of states, it cannot generate an infinite Burnside group.
\end{abstract}

\section{Mealy automata and the General Burnside problem}
The Burnside problem is a  famous, long-standing question in group theory. In 1902, Burnside  asked if a \emph{finitely generated} group whose all elements have finite order --henceforth called a \emph{Burnside group}-- is necessarily finite.

The question stayed open until Golod and Shafarevitch exhibit in 1964 an infinite group satisfying Burnside's conditions~\cite{golod, golod_shafarevich}, hence solving the general Burnside problem. In the early 60's, Glushkov suggested using automata to attack the Burnside problem~\cite{glushkov:automata}. Later, Aleshin~\cite{aleshin} in 1972 and then Grigorchuk~\cite{grigorchuk1} in 1980 gave simple examples of automata generating infinite Burnside groups. 
Over the years, automaton groups have been successfully used  to solve
several other group theoretical problems and conjectures such as
Atiyah, Day, Gromov or Milnor problems; the underlying automaton
structure can indeed be used to better understand the generated group.

It is remarkable that every known examples of infinite Burnside automaton groups   are generated by non reversible  Mealy automata, that is,  Mealy automata where the input letters do not all act like  permutations on the stateset. We conjecture that it is in fact
impossible for a reversible Mealy automaton to generate an infinite
Burnside group.
 Our past work with several co-authors has already given some partial
 results  to this end.
In \cite{GKP} it is proven that non bireversible Mealy automata
cannot generate  Burnside groups. For the whole class of reversible
automata, it has been proved in \cite{Kli13} that 2-state reversible
Mealy automata cannot generate  infinite Burnside groups. This result
has  later been extended in \cite{klimann_ps:3state} to 3-state
connected reversible automata. In this paper we generalize these
results to any connected revertible automaton with a prime number of states:

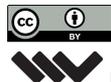
\begin{theorem*}
 A connected  invertible-reversible Mealy automaton of  prime size cannot generate an infinite Burnside group.
\end{theorem*}

Our proof is inspired by  the former work in the 3-state case of the second author with Picantin and Savchuk~\cite{klimann_ps:3state}. However the extension from 3 states to $p$ states, for any prime $p$, required the introduction of a new machinery. This constitutes the main part of our paper, see Section~\ref{sec-jungle}.\medskip

 The paper is organized as follows. In Section~\ref{sec-basic} we set
 up notations and recall useful facts on Mealy automata, automaton
 groups, and rooted trees. Then in Section~\ref{sec-cc} we link some
 characteristics  of an automaton group to the   connected components
 of the powers of the generating automaton. In Section~\ref{sec-otree}
 we introduce a tool developed in \cite{klimann_ps:3state}, the labeled orbit tree,
 that is used in Section~\ref{sec-jungle} to define our main tool, the
 \emph{jungle tree}. In this former section we also present some
 constructions and properties connected to this jungle tree. At last,
 in Section~\ref{sec-thm}, we gather our information and prove our
 main result. 
\section{Basic notions}\label{sec-basic}

\subsection{Groups generated by Mealy automata}
We first recall the formal definition of an automaton. A {\em (finite, deterministic, and complete) automaton} is a
triple
\(
\bigl( Q,\Sigma,\delta = (\delta_i\colon Q\rightarrow Q )_{i\in \Sigma} \bigr)
\),
where the \emph{stateset}~$Q$
and the \emph{alphabet}~$\Sigma$ are non-empty finite sets, and
the~\(\delta_i\) are functions.

A \emph{Mealy automaton} is a quadruple

\(( Q, \Sigma, \delta
,\rho
)\), 

such that \((Q,\Sigma,\delta)\) and~\((\Sigma,Q,\rho)\) are both
automata.
In other terms, a Mealy automaton is a complete, deterministic,
letter-to-letter transducer with the same input and output
alphabet. Its \emph{size} is the cardinal of its stateset.

The graphical representation of a Mealy automaton is
standard, see Figure~\ref{fig-bella}~left.

\newcommand{\trivaut}{\begin{tikzpicture}
\node[state] (a) {};
\path (a) edge[loop right] (a);
  \end{tikzpicture}
}
\newcommand{\bt}{\scalebox{.6}{\begin{tikzpicture}
\node[state] (a) {};
\node[state,below right of=a] (b) {};
\node[state,above right of=b] (c) {};
\path (a) edge[bend left] (c)
         (b) edge (a)
         (b) edge[loop right] (b)
         (c) edge[bend left] (a)
         (c) edge (b);
  \end{tikzpicture}}
}
\newcommand{\btcarrea}{\bt}
\newcommand{\btcarreb}[1][black]{\scalebox{.6}{\begin{tikzpicture}
\node[state] (ab) {};
\node[state,below right of=ab] (bc) {};
\node[state,above right of=bc] (cb) {};
\path (ab) edge[bend left] (cb)
         (bc) edge (ab)
         (cb) edge[bend left] (ab);
\node[state,right of=cb] (ac) {};
\node[state,below right of=ac] (ba) {};
\node[state,above right of=ba] (ca) {};
\path (ac) edge[bend left] (ca)
         (ba) edge (ac)
         (ca) edge[bend left] (ac);
\begin{scope}[#1]
\path (ab) edge[bend left=40] (ca)
         (cb) edge (ba)
         (bc) edge (ba)
         (ac) edge (cb)
         (ca) edge[bend left=60] (bc)
         (ba) edge[bend left=20] (bc);
\end{scope}
  \end{tikzpicture}}
}
\newcommand{\btcubea}{\btcarreb}
\newcommand{\btcubeb}[1][black]{\scalebox{.6}{\begin{tikzpicture}
\node[state] (aba) {};
\node[state,below right of=aba] (bca) {};
\node[state,above right of=bca] (cbc) {};
\path (aba) edge[bend left] (cbc)
         (cbc) edge[bend left] (aba);
\node[state,right of=cbc] (acb) {};
\node[state,below right of=acb] (bab) {};
\node[state,above right of=bab] (cac) {};
\path (bab) edge (acb)
         (cac) edge[bend left] (acb);
\begin{scope}
\path (aba) edge[bend left=40] (cac)
         (cbc) edge (bab)
         (cac) edge[bend left=60] (bca)
         (bab) edge[bend left=20] (bca);
\end{scope}
\begin{scope}[below=60]
\node[state] (abc) {};
\node[state,below right of=abc] (bcb) {};
\node[state,above right of=bcb] (cba) {};
\path (abc) edge[bend left] (cba)
         (cba) edge[bend left] (abc);
\node[state,right of=cba] (aca) {};
\node[state,below right of=aca] (bac) {};
\node[state,above right of=bac] (cab) {};
\path (bac) edge (aca)
         (cab) edge[bend left] (aca);
\path (abc) edge[bend left=40] (cab)
         (cba) edge (bac)
         (cab) edge[bend left=60] (bcb)
         (bac) edge[bend left=20] (bcb);
\begin{scope}[#1]
\path (acb) edge (cba)
         (acb) edge (cab)
         (bca) edge (abc)
         (bca) edge (bac)
         (aca) edge (cbc)
         (aca) edge (cac)
         (bcb) edge (aba)
         (bcb) edge (bab);
\end{scope}
\end{scope}
  \end{tikzpicture}}
}

\begin{figure}[h]
\centering

\begin{tikzpicture}
\begin{scope}[->,>=latex,every state/.style={minimum
      size=1pt,inner sep=1pt},level 1/.style={level distance=10mm},level 2/.style={sibling
      distance=4.5cm,level distance=.35cm},level 3/.style={sibling distance=1.8cm,level distance=2.8cm}],inner sep=0pt]
\node {\trivaut}
  child  {
  node {\bt}
     child { node {\btcarrea}
               child { node {\btcarrea}
                          edge from parent node[above left=2] {\(1\)}}
               child { node {\btcarreb[lightgray]}
                          edge from parent node[above right=2] {\(2\)}}
               edge from parent node[above left=2,near start] {\(1\)}
            }
     child { node {\btcarreb[lightgray]}
              child { node {\btcubea}
                         edge from parent node[above left=2] {\(1\)}}
              child[sibling distance=2.8cm] { node {\btcubeb[lightgray]}
                         edge from parent node[above right=2,very near end] {\(2\)}}
              edge from parent node[above right=2,very near start] {\(2\)}
            }
     edge from parent node[right=2] {\(3\)}
  }
;
\end{scope}
\begin{scope}[every state/.style={minimum
      size=3pt,inner sep=1pt},below=60,left=200,->,>=latex,node distance=1.5cm]
\node[state] (a) {};
\node[state,below right of=a] (b) {};
\node[state,above right of=b] (c) {};
\path (a) edge[bend left] node[above=-1] {\(0|1,1|0\)} (c)
         (b) edge node[left=2] {\(0|0\)} (a)
         (b) edge[loop right] node[below right=-3] {\(1|1\)}(b)
         (c) edge[bend left] node[above] {\(1|1\)} (a)
         (c) edge node[right=2] {\(0|0\)} (b);
\end{scope}
\end{tikzpicture}

\caption{The Bellaterra automaton~\(\aut{B}\)
and four levels of the orbit tree~\(\otree[B]\).}\label{fig-bella}

\end{figure}

A Mealy automaton \((Q,\Sigma,\delta, \rho)\) is
\emph{invertible\/} if the functions \(\rho_x\) are permutations of~\(\Sigma\)
and \emph{reversible\/} if the functions \(\delta_i\) are
permutations of~\(Q\).

In a Mealy automaton~\(\aut{A}=(Q,\Sigma, \delta, \rho)\), the sets~\(Q\)
and~\(\Sigma\) play dual roles. So we may consider the \emph{dual (Mealy)
automaton} defined by
\(
\dual{\aut{A}} = (\Sigma,Q, \rho, \delta)
\).
Obviously, a Mealy automaton is reversible if and only if its dual is
invertible.

An invertible Mealy automaton is \emph{bireversible} if it is
reversible (i.e. the input letters of the transitions act like
permutations on the stateset) and the output letters of the
transitions act like permutations on the stateset.

Let~\(\aut{A} = (Q,\Sigma, \delta,\rho)\) be a Mealy automaton.
Each state \(x\in Q\) defines a mapping from \(\Sigma^*\) into itself
recursively defined by:
\begin{equation*}
\forall i \in \Sigma, \ \forall \mot{s} \in \Sigma^*, \qquad
\rho_x(i\mot{s}) = \rho_x(i)\rho_{\delta_i(x)}(\mot{s}) \:.
\end{equation*}

The image of the empty word is itself.
The mapping~\(\rho_x\) for each $x\in Q$ is length-preserving and prefix-preserving.

We say that~\(\rho_x\) is the
function \emph{induced} by \(x\).
For~$\mot{x}=x_1\cdots x_n \in Q^n$ with~$n>0$, set
\(\rho_\mot{x}\colon\Sigma^* \rightarrow \Sigma^*, \rho_\mot{x} = \rho_{x_n}
\circ \cdots \circ \rho_{x_1} \:\).

Denote dually by~\(\delta_i\colon Q^*\rightarrow Q^*,
i\in \Sigma\), the functions induced by the states of~$\dz(\aut{A})$. 
For~$\mot{s}=s_1\cdots s_n
\in \Sigma^n$ with~$n>0$, set~\(\delta_\mot{s}\colon Q^* \rightarrow Q^*,
\ \delta_\mot{s} = \delta_{s_n}\circ \cdots \circ \delta_{s_1}\).

The semigroup of mappings from~$\Sigma^*$ to~$\Sigma^*$ generated by
$\{\rho_x, x\in Q\}$ is called the \emph{semigroup generated
  by~$\aut{A}$} and is denoted by~$\presm{\aut{A}}$.
When~\(\aut{A}\) is invertible, the functions induced by its states are
permutations on words of the same length and thus we may consider
the group of mappings from~$\Sigma^*$ to~$\Sigma^*$ generated by
$\{\rho_x, x\in Q\}$. This group is called the \emph{group generated
  by~$\aut{A}$} and is denoted by~$\pres{\aut{A}}$.

\subsection{Terminology on trees}\label{sec-trees}
Throughout this paper, we will use different sorts of labeled
trees. Here we set up some common terminology for all of them.

All our trees are rooted, \emph{i.e.} with a selected vertex called the \emph{root}. We will visualize the trees traditionally as growing down from the root. 
A \emph{path} is a (possibly infinite) sequence of adjacent edges without backtracking
from top to bottom. A path is said to be~\emph{initial} if it starts at the root of the tree.
A \emph{branch} is an infinite initial path.
The lead-off vertex of a non-empty path \(\mot{e}\) is denoted by~\(\top(\mot{e})\)
and its terminal vertex by~\(\bot(\mot{e})\) whenever the path is finite.

The \emph{level of a vertex} is its distance to the root and the
\emph{level of an edge} or \emph{a path} is the level of its initial vertex.

If the edges of a rooted tree are labeled
by elements of some finite set, the \emph{label} of a (possibly infinite)
path is the ordered sequence of labels of its edges.

Extending the notions of children, parents and descendent to the
edges, we will say that an edge \(f\) is the \emph{child} of an edge
\(e\) if \(\bot(e)=\top(f)\) (\emph{parent} being the converse notion,
and \emph{descendent} the transitive closure).

All along this article we will follow walks on some trees. A
\emph{walk} is just a path in a tree, which is build gradually. In
particular if \(\mot{e}\) is a finite path (or can identify one), to say
that it can be followed by \(f\) in some tree means that
\(\mot{e}f\) is (or identifies) also a
path in that tree.


\section{Connected components of the powers of an automaton}\label{sec-cc}
In this section we detail the basic properties of the connected
components of the powers of a reversible Mealy automaton, as it has
been done in~\cite{klimann_ps:3state}. The link
between these components is central in our construction.

Let \(\aut{A}=(Q,\Sigma,\delta,\rho)\) be a reversible Mealy automaton.

By reversibility,
all the connected components of its underlying graph are strongly connected.

Consider the powers of~\(\aut{A}\):
for~\(n>0\), its \emph{\(n\)-th power}~$\aut{A}^n$ is the Mealy automaton
\begin{equation*}
\aut{A}^n = \bigl( \ Q^n,\Sigma, (\delta_i\colon Q^n \rightarrow
Q^n)_{i\in \Sigma}, (\rho_{\mot{x}}\colon \Sigma \rightarrow \Sigma
)_{\mot{x}\in Q^n} \ \bigr)\enspace.
\end{equation*}
By convention, \(\aut{A}^0\) is the trivial automaton on the alphabet \(\Sigma\).

As \(\aut{A}\) is reversible, so are its powers
and the
connected components of~\(\aut{A}^n\) coincide with the orbits of the
action of~\(\pres{\dual{\aut{A}}}\) on~\(Q^n\).

Since~\(\aut{A}\) is reversible, there is a very particular connection
between the connected components of~\(\aut{A}^n\)
and those 
of~\(\aut{A}^{n+1}\) as highlighted
in~\cite{Kli13}. More precisely, 
take a
connected component~\(\aut{C}\) of some~\(\aut{A}^n\), 
and let~\(\mot{u}\in
Q^n\) (also written $|\mot{u}|=n$)  be a state of~\(\aut{C}\). 
Take also \(x\in Q\) a state of
$\aut{A}$, and let $\aut{D}$ be the connected component of $\aut{A}^{n+1}$
containing the state $\mot{u}x$.  Then, for any
state~\(\mot{v}\) of~\(\aut{C}\), there exists a state of~\(\aut{D}\)
prefixed with~\(\mot{v}\):
\[\exists\mot{s}\in\Sigma^*\mid \delta_{\mot{s}}(\mot{u}) =
\mot{v}\quad \text{and so}\quad  \delta_{\mot{s}}(\mot{u}x) =
\mot{v}\delta_{\rho_{\mot{u}}(\mot{s})}(x)\enspace.\]

Furthermore, if \(\mot{u}y\) is a state of~\(\aut{D}\), for some
state~\(y\in Q\) different from~\(x\), then \(\delta_{\mot{s}}(\mot{u}x)\)
and~\(\delta_{\mot{s}}(\mot{u}y)\) are two different states of~\(\aut{D}\)
prefixed with~\(\mot{v}\),
because of the reversibility of~\(\aut{A}^{n+1}\): the transition
function~\(\delta_{\rho_{\mot{u}}(\mot{s})}\) is a permutation.
Hence \(\aut{D}\) can be seen as consisting of several
copies of~\(\aut{C}\) and~\(\#\aut{C}\) divides~\(\#\aut{D}\). They
have the same size if and only if, for each state~\(\mot{u}\)
of~\(\aut{C}\) and any different states~\(x,y\in Q\), \(\mot{u}x\) and
\(\mot{u}y\) cannot simultaneously lie in \(\aut{D}\).

\smallskip

The connected components of the powers of a Mealy automaton and the
 finiteness of the generated group or 
 of a monogenic subgroup are closely related,
as shown in the following propositions (obtained independently in~\cite{klimann_ps:3state,dangeli_r:geometric}).

\begin{proposition}\label{prop-bounded-cc}

A reversible Mealy automaton generates a finite group if
and only if the connected components of its powers have bounded size.
\end{proposition}

\begin{proposition}\label{prop-finite}
Let $\aut{A}=(Q,\Sigma,\delta,\rho)$ be an invertible-reversible Mealy automaton and let~\(\mot{u}\in Q^+\) be a non-empty word. The following conditions are
equivalent:
\begin{enumerate}[(i)]
\item \(\rho_{\mot{u}}\) has finite order,\label{i1}
\item the sizes of the connected components of~\((\mot{u}^n)_{n\in
  \N}\) are bounded,\label{i2}
\item there exists a word \(\mot{v}\) such that the sizes of the
  connected components of~\((\mot{vu}^n)_{n\in \N}\) are
  bounded,\label{i3}
\item for any word \(\mot{v}\), the sizes of the
  connected components of~\((\mot{vu}^n)_{n\in \N}\) are
  bounded.\label{i4}
\end{enumerate}
\end{proposition}

\section{The Labeled Orbit Tree and the Order
  Problem}\label{sec-otree}
Most of the notions of this section have been introduced in~\cite{klimann_ps:3state}.

We build a tree capturing the links between the
connected components of consecutive powers of a reversible Mealy
automaton. See an example in Figure~\ref{fig-bella}. As recalled at
the end of this section, the existence of elements of infinite order
in the semigroup generated by an invertible-reversible automaton is closely related to some
path property of this tree.

Let \(\aut{A}=(Q,\Sigma,\delta,\rho)\) be a reversible Mealy automaton.
Consider the tree with vertices
the connected components of the powers of~\(\aut{A}\), and the
incidence relation built by adding an element of~\(Q\): for any~\(n\geq 0\),
the connected component of~\(\mot{u}\in Q^n\) is linked
to the connected component(s) of~\(\mot{u}x\), for any~\(x\in
Q\). This tree is called the \emph{orbit tree} of~\(\dual{\aut{A}}\)~\cite{gawron_ns:conjugation,grigorch_s:ergodic_decomposition}.
It can be seen as the
quotient of the tree~\(Q^*\) under the
action of the group~\(\pres{\dual{\aut{A}}}\).

We label any
edge \(\aut{C}\to\aut{D}\) of the orbit tree by the ratio
\(\frac{\#\aut{D}}{\#\aut{C}}\), which is always an integer by the
reversibility of~\(\aut{A}\). We call this labeled tree the
\emph{labeled orbit tree} of~\(\dual{\aut{A}}\)~\cite{klimann_ps:3state}.
We denote by~\(\otree[A]\) the labeled orbit tree of~\(\dual{\aut{A}}\).
Note that for each vertex of \(\otree[A]\), the sum of the labels of
all edges going down from this vertex always equals to \(\#Q\), the
size of~\(\aut{A}\).

Each vertex of \(\otree[A]\) is labeled by a connected automaton with
stateset in \(Q^n\), where \(n\) is the level of this vertex in the tree. By a
minor abuse, we can consider that each vertex is labeled by a finite
language in~\(Q^n\), or even by a word in~\(Q^n\).

Let \(\mot{u}\) be a (possibly infinite) word over \(Q\). The \emph{path
of}~\(\mot{u}\) in the orbit tree~\(\otree[A]\) is the unique
initial path going from the root through the connected
components of the prefixes of~\(\mot{u}\); \(\mot{u}\) can be called
\emph{a representative} of this initial path (we can say equivalently
that this path is \emph{represented} by~\(\mot{u}\) or that the word
\(\mot{u}\) \emph{represents} the path).

\begin{definition}
Let~\(\aut{A}\) be a reversible Mealy automaton
and~$\mathfrak{s}$ be a subtree of \(\otree[A]\). An
\emph{\(\mathfrak{s}\)-word} is a word in \(Q^*\cup Q^{\infty}\)
representing an initial path of \(\mathfrak{s}\). A \emph{cyclic
  \(\mathfrak{s}\)-word} is a word in \(Q^*\) whose all powers are
\(\mathfrak{s}\)-words (equivalently, it is an \(\mathfrak{s}\)-word
viewed as a cyclic word).
\end{definition}

The structure of an orbit tree is not arbitrary and it is possible to
identify some similarities inside this tree.

\begin{definition}
Let \(e\) and~\(f\) be two edges in the orbit tree~\(\otree[A]\).
We say that~\(e\) \emph{is liftable to}~\(f\) if each word of~\(\bot(e)\)
admits some word of~\(\bot(f)\) as a suffix.
\end{definition}

Consider \(\mot{u}$ in~$\top(e)\) and its suffix~\(\mot{v}\) in \(\top(f)\):
any state~\(x\in Q\) such that \(\mot{u}x\in\bot(e)\) satisfies~\(\mot{v}x\in \bot(f)\).
Informally, ``\(e\) liftable to \(f\)'' means
that what can happen after~\(\top(e)\) by following~\(e\) can also
happen after~\(\top(f)\) by following~\(f\). This condition is equivalent to a weaker one:

\begin{lemma}\label{lem-lift-one}
Let~\(\aut{A}\) be a reversible Mealy automaton, and let~\(e\) and \(f\) be two edges in the orbit tree~\(\otree[A]\).
If there exists a word of~$\bot(e)$ which admits a word of~$\bot(f)$ as suffix,
then $e$ is liftable to~$f$.
\end{lemma}

Obviously if \(e\) is liftable to~\(f\), then \(f\) is closer to the root of the orbit tree.
The fact that an edge is liftable to another one reflects a deeper relation stated below.
The following lemma is one of the key observations. 

\begin{lemma}\label{lem-liftable}
Let \(e\) and~\(f\) be two edges in the orbit tree~\(\otree[A]\). If \(e\) is
liftable to~\(f\), then the label of~\(e\) is less than or equal to the label
of~\(f\).
\end{lemma}

The notions of children of an edge and of being liftable to it are not
linked, but it is interesting to consider their intersection.

\begin{definition}
Let \(e\) and \(f\) be two edges in an orbit tree: \(f\) is a
\emph{legitimate child} of \(e\) if \(e\) is its parent and \(f\) is
liftable to \(e\).
\end{definition}

It is straightforward to notice that the label of an edge in an orbit
tree is equal to the sum of the labels of its legitimate children.

\bigskip

The notion of liftability can be generalized to paths:
\begin{definition}
Let \(\mot{e}=(e_i)_{i\in I}\)
and~\(\mot{f}=(f_i)_{i\in I}\) be two paths of the same (possibly
infinite) length in the orbit tree~\(\otree[A]\). The path \(\mot{e}\) \emph{is liftable to} the path
\(\mot{f}\) if, for any \(i\in I\), the edge~\(e_i\) is liftable to the edge \(f_i\).
\end{definition}

\begin{definition}\label{def-self-liftable}
Let~\(\aut{A}\) be a bireversible Mealy automaton
and~$\mathfrak{s}$ be a (possibly infinite) path or subtree of 
\(\otree[A]\).
For~$k>0$,
$\mathfrak{s}$ is \emph{$k$-self-liftable}
whenever any path in~$\mathfrak{s}$ starting at level~$i+k$ is
liftable to a path in~$\mathfrak{s}$ starting at level~$i$, for any~$i\geq 0$.
A path or a subtree is \emph{self-liftable} if it is \(k\)-self-liftable for some \(k>0\).
\end{definition}

\label{sec-char}

We recall here a characterization of the existence of elements of infinite order
in the semigroup generated by a reversible Mealy automaton~$\aut{A}$
in terms of path properties of the associated orbit tree~$\otree[A]$~\cite{klimann_ps:3state}.

\begin{definition}
Any branch labeled by a word not suffixed by~$1^\omega$ is called~\emph{active}.
\end{definition}

\begin{theorem}\label{thm-char}\cite{klimann_ps:3state}
The semigroup
generated by an invertible-reversible automaton~$\aut{A}$ admits elements of infinite order
if and only if the orbit tree~\(\otree[A]\) admits an active self-liftable branch.
\end{theorem}

\section{Jungle Trees}
\label{sec-jungle}
Our main result being known for non bireversible
automata~\cite{GKP}, we focus on the bireversible case. All the tools
introduced in this section are new. They are used to get rid of the
particularity of the stateset of size~3 in~\cite{klimann_ps:3state}.

Let $\auta=(Q,\Sigma, \delta, \rho)$ be a connected bireversible Mealy
automaton 
with no active self-liftable branch. From
Theorem~\ref{thm-char}, all the elements of the semigroup
\(\presm{\aut{A}}\) have finite order. In this section we introduce
the tools to prove that such an automaton generates a finite group
(Theorem~\ref{thm:main}).

\subsection{Jungle trees and stems}
We focus on some particular subtrees of $\ota$: 
\begin{definition}
Let \(\mot{e}\) be a finite initial 1-self-liftable path such that:
\begin{itemize}
\item $\bot(\mot{e})$ has at least two legitimate children;
\item every legitimate child of $\bot(\mot{e})$ has label 1.
\end{itemize}
 The \emph{jungle tree} \(\jungle(\mot{e})\) of \(\mot{e}\) is the subtree
of \(\otree[A]\) build as follows:
\begin{itemize}
\item it contains the path \(\mot{e}\) --- its \emph{trunk};
\item it contains the regular tree rooted by \(\bot(\mot{e})\), and
  formed by all the edges which are descendant of
  \(\bot(\mot{e})\) and  liftable to the lowest (\ie the last) edge of $\mot{e}$.
\end{itemize}
The \emph{arity} of this jungle tree is the number of legitimate
children of $\bot(\mot{e})$. Since every legitimate child has label $1$, it is
also the label of the last edge of $\mot{e}$.

Words in $\bot(\mot{e})$ are called \emph{stems}. They have all the
same length which is the length of the trunk of \(\jungle(\mot{e})\).

A tree is a \emph{jungle tree} if it is the jungle tree of some
finite initial 1-self-liftable path.
\end{definition}

Graphically, a jungle tree starts with a linear part whose labels
decrease (its trunk) and eventually ends as a regular tree with all
labels~1. Any jungle tree is 1-self-liftable.

Note that: (i) there exists at least one jungle tree, since
\(\aut{A}\) has no active self-liftable branch by hypothesis; (ii) there are
finitely many jungle trees.

{\bf From now on, \(\jungle\) denotes a jungle tree of \(\aut{A}\),
  whose trunk has length~\(n\).}

As shown below, any cyclic \(\jungle\)-words has finite order.

\begin{remark}
The
existence of cyclic \(\jungle\)-words is ensured by the simple fact
that any \(\jungle\)-word of length~\(n\times(1+\#Q^{n})\)
admits a cyclic \(\jungle\)-word as a factor.
\end{remark}

\begin{proposition}\label{prop-const-order}
Every
cyclic \(\jungle\)-word induces an action of finite order, bounded by a
uniform constant depending on~\(\jungle\).
\end{proposition}

\begin{proof}
Let \(\mot{u}\) be a cyclic \(\jungle\)-word, then, for any integer~\(k>n\), \(\mot{u}^k\) is a \(\jungle\)-word. By the definition of a
jungle tree, the label of the path of
\(\mot{u}^{\omega}\) is ultimately~1 and,
by~Proposition~\ref{prop-bounded-cc}, the action induced by~\(\mot{u}\)
has finite order, bounded by a constant which depends on the connected
component at the end of the trunk of~\(\jungle\).
\end{proof}

Because of the self-liftability of \(\jungle\), any factor of a
\(\jungle\)-word is itself a \(\jungle\)-word. Hence any factor of
length~\(n\) of a \(\jungle\)-word is a stem. And by the construction of 
\(\jungle\), the end of its trunk has only one vertex whose label is hence a
connected component, and all the stems are states of that same connected component.

\begin{myremark}\label{rem:canfollow}
If \(\mot{uv}\) is a \(\jungle\)-word, with \(|\mot{v}|\geq n\), what
can follow \(\mot{uv}\) in \(\jungle\) is independant
from~\(\mot{u}\). In particular, if \(\mot{vw}\) is also a
\(\jungle\)-word, then so is \(\mot{uvw}\).
\end{myremark}

\begin{definition}
Let \(\jungle\) be a jungle tree of trunk of length~\(n\). A
\emph{liana covering up \(\jungle\)} is a language of
\(\jungle\)-words, of the form \(\Li{\mot{}w}\), where
\(\mot{w}\in Q^n\) is a stem, and \(L_{\mot{w}}\subseteq
Q^*\cup Q^{\infty}\) is a prefix-preserving language which, seen as a
tree, is regular of the same arity as \(\jungle\).

Each vertex of \(\jungle\) has exactly one representative in
\(\Li{\mot{w}}\). For each stem $\mot{w}$ there is exactly one suitable $L_{\mot{w}}$.
\end{definition}

\begin{myremark}\label{rem:new-liana}
Let \(\mot{w}L_{\mot{w}}\) be a liana covering up a jungle tree
\(\jungle\) and \(\mot{uv}\) be a finite \(\jungle\)-word such that
\(|\mot{v}| = n\): if \(L_{\mot{v}}\) is
the greatest language such that \(\mot{uv}L_{\mot{v}}\subseteq
\mot{w}L_{\mot{w}}\), then \(\mot{v}L_{\mot{v}}\) is also a liana
covering up~\(\jungle\).
\end{myremark}

\medskip

In what follows, we try to better understand the stucture of jungle
trees and lianas. Let  \(S=\mot{s}L_{\mot{s}}\) be a liana covering up
\(\jungle\) (\(\mot{s}\in Q^n\)). Our goal is to prove the following
result:

\begin{theorem}\label{thm:walk-on-liana}\label{thm:ubiquity}
Let \(\mot{u}\) be a factor in \(S\). Then \(\mot{u}\) has the
following property:
\phantomsection
\begin{equation}\label{prop:W}
\text{If }\mot{uv}\in Q^*\text{ is a factor in }S\text{, then
}\mot{u}\text{  exists further in S.}\tag{{\bf Ubiquity}}
\end{equation}
More formally: if \(\mot{tuv}\in S\), there exists \(\mot{w}\in Q^*\)
such that \(\mot{tuvwu}\in S\).
\end{theorem}

The graphical sense of this theorem is that if you are walking on a
\(\jungle\)-word and you have already seen some factor, you can find
eventually this same factor.

\begin{proof}
First, remember that if \(\mot{u}\) is a stem (\ie \(\mot{u}\) is a factor
in \(S\) of length~\(n\)), what can follow \(\mot{u}\) (in~\(S\)) does not depend
either of the choice of the liana (as long as you are in a liana
covering up the same jungle tree), or of the location of \(\mot{u}\) in
this liana. Hence it is sufficient to prove
Theorem~\ref{thm:walk-on-liana} for \(|\mot{u}|=n\).

\smallskip

We start by proving that
there is at least one stem \(\mot{u}_0\) with
Property~\eqref{prop:W}. To obtain this word, we travel along
\(S=\mot{s}L_{\mot{s}}\) in the following way, 
starting with \(\mot{u}_0=\mot{s}\):
\begin{itemize}
\item if \(\mot{u}_0\) answers to the question, our journey is over;
\item otherwise, at the end of \(\mot{u}_0\) we follow some finite path
  such that \(\mot{u}_0\) does not exist anymore after this path; then we
  replace \(\mot{u}_0\) by the next word of length~\(n\) in~\(S\), and
  back to the previous step.
\end{itemize}
Since \(S\) is infinite but has a finite arity and a finite number of factors of
length~\(n\), the previous algorithm ends returning a stem
\(\mot{u}_0\) satisfying Property~\eqref{prop:W}. By
Remark~\ref{rem:new-liana}, the
jungle tree \(\jungle\) is covered up by a liana of the form
\(\mot{u}_0L_{\mot{u}_0}\).

\medskip

The extension of Property~\eqref{prop:W} to other words is illustrated
by Figure~\ref{fig:extW}. 
Let~\(\mot{u}\mot{v}\) be a factor in \(S\), with
\(|\mot{u}|=n\). In particular \(\mot{u}\) is a stem, hence
\(\mot{u}_0\) and \(\mot{u}\) are states of
the same connected component, and there exists a path in this
component from \(\mot{u}_0\) to \(\mot{u}\), say by the action of some
\(\mot{a}\in \Sigma^*\). The automaton being reversible, \(\mot{v}\) is the image of some
\(\mot{v}_0\in Q^{|\mot{v}|}\) by  \(\mot{b}=\delta_{\mot{u}_0}(\mot{a})\) (see
Figure~\ref{fig:extW}).

Now, we know that on the left part of Figure~\ref{fig:extW}  we can find eventually
\(\mot{u}_0\), after some \(\mot{w}_0\) (because \(\mot{u}_0\) has
Property~\eqref{prop:W}). 
And
by the invertibility of the automaton, there exists some power
\(\alpha\) of \(\mot{u}_0\mot{v}_0\mot{w}_0\) which stabilizes
\(\mot{a}\):

\begin{figure}[H]
\centering
\[\begin{array}{ccc}
 &  \mot{a} & \\
\mot{u}_0 & \lacroix & \mot{u}\\
 & \mot{b} &\\
\mot{v}_0 & \lacroix & \mot{v}\\
 & \dots &\\
\mot{w}_0 & \lacroix & \dots\\
 & \dots &\\
(\mot{u}_0\mot{v}_0\mot{w}_0)^{\alpha-1}& \lacroix & \dots\\
 &  \mot{a} & \\
\mot{u}_0 & \lacroix & \mot{u}\\
\end{array}\]
\caption{Extension of Property~\eqref{prop:W} from \(\mot{u}_0\) to \(\mot{u}\).}\label{fig:extW}

\end{figure}

Hence \(\mot{u}\) can be seen again eventually and it has Property~\eqref{prop:W}.
\end{proof}

\begin{myremark}
Note that, from Theorem~\ref{thm:walk-on-liana}, if $\mot{u}$, $\mot{v}$ are two stems such that $\mot{v}$ is a factor of some word  in $\Li{\mot{u}}$, then  $\mot{u}$ is a factor of some word in  $ \Li{\mot{v}}$.
\end{myremark}

\subsection{An equivalence on stems}
Remember that \(\aut{A}=(Q,\Sigma, \delta, \rho)\) is a connected bireversible Mealy
automaton such that \(\ota\) has no active self-liftable branch (and as a
consequence all the elements of the semigroup \(\presm{\aut{A}}\) have
finite order).
Let \(\jungle\) be a jungle tree of \(\ota\) with trunk of
length $n$. All the stems considered from now on are stems of \(\jungle\).

In this subsection we prove several properties for the stems
of the jungle tree \(\jungle\). Stems are used then in Section~\ref{sec:main} to
build a \(\jungle\)-word inducing the same action than some given word.

\medskip

Let us first introduce an equivalence relation on the set of stems.

\begin{definition}
Let $\mot{u}$, $\mot{v}$ be two stems. We say that $\mot{u}$ is
\emph{equivalent} to $\mot{v}$, denoted by $\mot{u}\sim \mot{v}$,
whenever there exists $\mot{s} \in Q^*$ such that $\mot{usv}$ is a
\(\jungle\)-word  and $\rho_{\mot{us}}$ acts like the identity on
$\Sigma^*$.
\end{definition}

\begin{lemma}
The relation $\sim$ is an equivalence relation on stems.
\end{lemma}
\begin{proof}
Let \(\mot{u}\), \(\mot{v}\), and \(\mot{w}\) be three stems.
\begin{description}
\item[transitivity]  Suppose that $\mot{u}\sim \mot{v}$  and \(\mot{v} \sim
  \mot{w}\): there exists \(\mot{s},\mot{t}\in Q^*\) such that
  \(\mot{usv}\) and \(\mot{vtw}\) are \(\jungle\)-words, and \(\rho_{\mot{us}}\)
  and \(\rho_{\mot{vt}}\) act like the identity. As \(\mot{v}\) is a stem, we
  obtain by Remark~\ref{rem:canfollow} that \(\mot{usvtw}\) is a \(\jungle\)-word, and
  \(\rho_{\mot{usvt}}\) acts like the identity, so \(\mot{u}\sim\mot{w}\).
\item[reflexivity] From Theorem~\ref{thm:ubiquity}, there exists
  \(\mot{s}\in Q^*\) such that \(\mot{usu}\) is a \(\jungle\)-word (in
  fact from Theorem~\ref{thm:ubiquity} one can even chose the beginning
  of \(\mot{s}\), as long as we keep a \(\jungle\)-word). As
  \(\mot{u}\) is a stem, \(\mot{usus}\) is also a \(\jungle\)-word,
  and so are all the powers of \(\mot{us}\). Now, by hypothesis and Theorem~\ref{thm-char},
  \(\mot{us}\) is of finite order, say \(\alpha\):
  \(\mot{u}(\mot{su})^{\alpha-1}\mot{su}\) is a \(\jungle\)-word and
  \(\rho_{\mot{u}(\mot{su})^{\alpha-1}\mot{s}} = \rho_{(\mot{us})^{\alpha}}\) acts
  like the identity.
\item[symmetry] Suppose that \(\mot{u}\sim \mot{v}\): there exists
  \(\mot{s}\in Q^*\) such that \(\mot{usv}\) is a \(\jungle\)-word and
  \(\rho_{\mot{us}}\) acts like the identity. From the reflexivity proof,
  there exists \(\mot{t}\in Q^*\) such that \(\mot{usvtu}\) is a
  \(\jungle\)-word and \(\rho_{\mot{usvt}}\) acts like the
  identity. Hence
  \(\mot{vtu}\) is a \(\jungle\)-word and \(\rho_{\mot{vt}}\) acts like the
  identity, which proves the symmetry.
\end{description}

\end{proof}

Note that from reflexivity of~\(\sim\) and Theorem~\ref{thm:ubiquity},
if \(\mot{u}\) and \(\mot{v}\) are equivalent stems and \(\mot{uw}\)
is a \(\jungle\)-word for some \(\mot{w}\in Q^*\), then there exists
\(\mot{s}\in Q^*\) such that \(\mot{uwsv}\) is a \(\jungle\)-word and
\(\rho_{\mot{uws}}\) acts like the identity. So not only \(\mot{v}\) can be
reached from \(\mot{u}\) by producing first the identity, but even if you
walk in \(\jungle\) after reading \(\mot{u}\), you can still reach
\(\mot{v}\) and produce first the identity.

\bigskip

We can now consider the equivalence classes induced by~\(\sim\). The
aim of this subsection is to prove that if \(\aut{A}\) has a prime
size, then for a given state \(q\) there is in each
\(\sim\)-class a stem with prefix~\(q\) (Theorem~\ref{thm:Qprefix}).

\begin{proposition}\label{prop:samesize}
All the equivalence classes of $\sim$ have the same size.
\end{proposition}
\begin{proof}
Let \(\mot{u}_0\) and \(\mot{v}_0\) be two 
stems of \(\jungle\): they are states of the same connected component and
so there exists \(\mot{a}\in\Sigma^*\) such that
\(\delta_{\mot{a}}(\mot{u}_0)=\mot{v}_0\). Denote by
\(\{\mot{u}_0,\dots,\mot{u}_k\}\) the \(\sim\)-class of
\(\mot{u}_0\): for any \(i\), \(1\leq i\leq k\), there exists
\(\mot{s}_i\in Q^*\) such that \(\mot{u}_0\mot{s}_i\mot{u}_i\) is a
\(\jungle\)-word and \(\rho_{\mot{u}_0\mot{s}_i}\) acts like the
identity. Define the
words \(\mot{v}_i\in Q^{|\mot{u}_i|}\) and \(\mot{t}_i\in Q^{|\mot{s}_i|}\) in the following way:
\(\delta_{\mot{a}}(\mot{u}_0\mot{s}_i\mot{u}_i)=\mot{v}_0\mot{t}_i\mot{v}_i\). Note that $\mot{v}_0\mot{t}_i\mot{v}_i$ is also a \(\jungle\)-word: any factor of size $n$ of  $\mot{v}_0\mot{t}_i\mot{v}_i$ is the image of a stem (the corresponding factor in $\mot{u}_0\mot{s}_i\mot{u}_i$) and therefore belongs to the connected component of $\mot{u}_0$ and $\mot{v}_0$, hence every prefix of $\mot{v}_0\mot{t}_i\mot{v}_i$ is on a 1-self-liftable path.  Now \(\rho_{\mot{v}_0\mot{t}_i}\) acts like the
identity by the reversibility of~\(\aut{A}\), so \(\mot{v}_i\) is
$\sim$-equivalent to \(\mot{v}_0\). Furthermore, as \(\rho_{\mot{u}_0\mot{s}_i}\)
acts like the identity, we know that
\(\mot{v}_i=\delta_{\mot{a}}(\mot{u}_i)\), and all the \(\mot{v}_i\)
are different.
\end{proof}

\subsection{Combinatorial properties of stems}

We now states several combinatorial properties for stems. Let
$k_1,k_2,\hdots ,k_n$ be the labels, from root to $\bot(\mot{e})$, of
the jungle tree $\jungle=\jungle(\mot{e})$. Recall that, since $\auta$
is connected, $k_1=p$ and by construction of the jungle tree $k_n \geq
2$.

 First if we consider  no restriction then we can directly count stems by looking to the labels of the trunk:
 
\begin{lemma}\label{lem:sizestem}
The number of stem with a given prefix depends only on length $i$ of the prefix and is $k_{i+1}k_{i+2}\hdots k_n$.
\end{lemma}

We are now interested in two somehow dual questions. Fix a
\(\jungle\)-word \(\mot{u}\) of length less than~\(n\): (i) if
\(\mot{u}\) is the prefix of a stem in some \(\sim\)-class
\(\equivclass\), in how many way can \(\mot{u}\) be completed in
\(\equivclass\) (Proposition~\ref{prop:numsufeq})? (ii) in how many
\(\sim\)-classes is \(\mot{u}\) the prefix of a stem
(Corollary~\ref{cor:fixprefix})?

\begin{proposition}\label{prop:numsufeq}

Fix some \(\jungle\)-word \(\mot{u}\) of length less than \(n\), a
\(\sim\)-class \(\equivclass\) of stems including an element with
prefix \(\mot{u}\), and
some integer \(k\) such that \(|\mot{u}|+k\leq n\).
The number of \(\mot{v}\in Q^k\) such that \(\mot{uv}\) is a prefix of a stem
of \(\equivclass\) depends only on \(|\mot{u}|\) and \(k\).
\end{proposition}

\begin{proof}
By the same argument as in the proof of Proposition~\ref{prop:samesize}.
\end{proof}

Let \(\mot{u}\in Q^*\) be a prefix of a stem in some \(\sim\)-class
\(\equivclass\). Denote by \(\Sufeq(|\mot{u}|+1)\) the cardinal of the set
\(\{q\in Q\mid \mot{u}q\text{ is a prefix of some stem in
}\equivclass\}\) (from Lemma~\ref{prop:numsufeq} it depends only of \(|\mot{u}|\)
and so it is correctly defined).

In order to obtain a minimal bound on the size of a $\sim$-class,
we introduce another equivalence relation between stems which is finer
than~\(\sim\), as proved in Lemma~\ref{lem:sibleq}: 
\begin{definition}
Let $\mot{u},\mot{v}$ be two stems. Define the relation $\mot{u}\wedge_0
\mot{v}$ whenever there exists a stem $\mot{s}$ such that both
$\mot{su}$ and $\mot{sv}$ are \(\jungle\)-words. The equivalence relation $\wedge$
is defined as the transitive closure of $\wedge_0$.

\end{definition}

\begin{lemma}\label{lem:sibleq}
The relation \(\wedge\) is finer than the relation \(\sim\): $\mot{u}\wedge \mot{v} \Rightarrow \mot{u}\sim\mot{v}$.
\end{lemma}
\begin{proof}

By transitivity it is enough to prove that: $\mot{u} \wedge_0
\mot{v} \Rightarrow \mot{u}\sim \mot{v}$. Let \(\mot{u}\) and
\(\mot{v}\) be two stems such that \(\mot{u} \wedge_0 \mot{v}\):
there exists a stem \(\mot{s}\) such that \(\mot{su}\) and
\(\mot{sv}\) are \(\jungle\)-words. From Theorem~\ref{thm:ubiquity},
there exists a word \(\mot{w}\in Q^*\) such that \(\mot{uws}\) is a
\(\jungle\)-word. As \(\mot{u}\) and \(\mot{s}\) are stems, and
\(\mot{su}\) is a \(\jungle\)-word, \((\mot{uws})^2\) is 
a \(\jungle\)-word, and so are all the powers of \(\mot{uws}\). Now, by
hypothesis and Theorem~\ref{thm-char}, the word \(\mot{uws}\) has finite order, say \(\alpha\):
\((\mot{uws})^{\alpha}\mot{v}\) is a \(\jungle\)-word and
\(\rho_{(\mot{uws})^{\alpha}}\) acts like the identity.

\end{proof}

\begin{corollary}\label{cor:suf}
For any \(i\), \(\Sufeq(i)\geq 2\).
\end{corollary}

\begin{proof}
For a fixed stem \(\mot{u}\), the language of words in
  \(\wedge_0\)-relation with \(\mot{u}\), seen as a tree, has the
same arity than \(\jungle\), and so, by Lemma~\ref{lem:sibleq}, for any $i$,
\(\Sufeq(i)\) is greater than or equal to the arity of~\(\jungle\).
\end{proof}

\begin{proposition}\label{prop:fixprefix}

Fix a \(\jungle\)-word \(\mot{u}\) of length less than~\(n\).
The number of stems prefixed by \(\mot{u}\) in a \(\sim\)-class is
either~\(0\) or depends only on \(|\mot{u}|\).
\end{proposition}

\begin{proof}

By the same argument as in the proof of Proposition~\ref{prop:samesize}.
\end{proof}

From Propositions~\ref{prop:numsufeq} and~\ref{prop:fixprefix} we obtain:
\begin{corollary}\label{cor:fixprefix}
Fix a \(\jungle\)-word \(\mot{u}\) of length less than~\(n\).
The number of \(\sim\)-classes where \(\mot{u}\) is the prefix of
some stem depends only on \(|\mot{u}|\).
\end{corollary} 

Denote by $\Preeq(|\mot{u}|+1)$ the number of \(\sim\)-classes
containing a stem prefixed by \(\mot{u}\) (it is correctly define by
Corollary~\ref{cor:fixprefix}). 
\medskip

We can now prove the main result of this section:
\begin{theorem}\label{thm:Qprefix}
Let $\auta$ be a connected bireversible Mealy automaton of prime size
and without any active self-liftable branch. The set of
states which appear as first letter of a stem in a fixed \(\sim\)-class is the whole stateset.
\end{theorem} 

\begin{proof}
Suppose \(\aut{A}=(Q,\Sigma,\delta,\rho)\) has prime size \(p\), and let
\(\jungle\) be a jungle tree of \(\ota\) whose 
trunk~\(\mot{e}\) has length~\(n\). We denote by \(k_1\),\dots ,\(k_n\)
the label of the edges of \(\mot{e}\) (from top to bottom). By the
connectivity of \(\aut{A}\), \(k_1=p\).

Let \(\equivclass\) be a \(\sim\)-class of stems for \(\jungle\)
and \(\mot{u}\in Q^*\) of length \(i\leq n\) be the prefix of some
stem in~\(\equivclass\). Consider all the stems in~\(\equivclass\)
with prefix~\(\mot{u}\).

From Lemma~\ref{lem:sizestem}, the number of stems of \(\jungle\)  prefixed by \(\mot{u}\) is
$k_{i+1}\times k_{i+2}\times \hdots\times k_n$. On the other hand, it is
also the number of stems with prefix~\(\mot{u}\) in \(\equivclass\),
\ie \(\Sufeq(i+1)\times\dots \times\Sufeq(n)\), 
times the number of \(\sim\)-classes which has a stem prefixed
by~\(\mot{u}\), \ie 
 \(\Preeq(i+1)\times\dots \times\Preeq(n)\). Hence
\[k_{i+1}\times k_{i+2}\times \hdots \times k_n = \Sufeq(i+1)\times \Preeq(i+1) \times\Sufeq(i+2)\times \Preeq(i+2) \times\hdots \times  \Sufeq(n)\times \Preeq(n) \:.\]

It is straightforward that $k_n = \Preeq(n)\times \Sufeq(n) $ and  by
induction \(\Preeq(\ell) \times \Sufeq(\ell)= k_{\ell}\) for all \(\ell\). In
particular for $\ell=1$, we get that \(\Sufeq(1)\) devides \(k_1\). Since
$k_1=p$ and, from Corollary~\ref{cor:suf}, $\Sufeq(1)\geq 2$, we
obtain then \(\Sufeq(1)=p\).
\end{proof}

\begin{corollary}\label{cor-id-follow}
Let \(\aut{A}=(Q,\Sigma,\delta,\rho)\) be a connected bireversible Mealy automaton of prime size, with no active self-liftable branch. Let 
\(\jungle\) be a jungle tree of \(\otree[A]\) and  \(\mot{u}\)  some
\(\jungle\)-word. Then for any state \(x\in Q\), there exists
\(\mot{w}\in Q^*\) such that \(\mot{uw}x\) is a \(\jungle\)-word and
\(\rho_{\mot{w}}\) acts like the identity of~\(\Sigma^*\).
\end{corollary}

\begin{proof}
Let \(\mot{s}\) be a stem such that \(\mot{us}\) is a
\(\jungle\)-word: there exists a stem \(\mot{x}\) with first
letter~\(x\) in the \(\sim\)-class of \(\mot{s}\), from
Theorem~\ref{thm:Qprefix}, \ie there exists \(\mot{v}\in Q^*\) such
that \(\mot{svx}\) is a \(\jungle\)-word and \(\rho_{\mot{sv}}\) acts like
the identity of~\(\Sigma^*\). Conclusion comes from
Remark~\ref{rem:canfollow}.
\end{proof}

Note that in the previous corollary, the word \(\mot{u}\) can be empty.

\section{Proof of the main theorem}\label{sec-thm}\label{sec:main}

We now have all elements to prove our main result.

\begin{theorem}\label{thm:main}
 A connected  invertible-reversible Mealy automaton of  prime size cannot generate an infinite Burnside group.
\end{theorem}

\begin{proof}
Let $\auta$ be a connected invertible-reversible Mealy automaton of
prime size.
If $\auta$ is not bireversible we can apply \cite{AKLMP12,GKP}
and we get that, on one hand, $\gauta$ is necessarily infinite, but on
the other hand, it cannot be Burnside.
If $\auta $ is bireversible and $\ota$ has an active self-liftable branch,
then $\gauta$ has an element of infinite order by
Theorem~\ref{thm-char}.

Therefore we can assume that  $\auta $ is bireversible and $\ota$ has
no active self-liftable branch. Let us show that $\gauta$ is finite. Let
\(\jungle\) be some jungle tree of \(\ota\). As in
\cite{klimann_ps:3state} we prove that for any word \(\mot{u}\in
Q^*\), $\rho_{\mot{u}}$ has some uniformly bounded power which acts like some cyclic $\jungle$-word.

Let $\mot{u} \in Q^*$. We prove by induction that any prefix $\mot{u}$
induces the same action as some $\jungle$-word. It is obviously true
for the empty prefix. Fix some $k < |\mot{u}|$ and suppose that the
prefix $\mot{v}$ of length $k$ of $\mot{u}$ induces the same action as
some $\jungle$-word $\mot{s}$. Let $x \in Q$ be the $(k+1)$-th letter
of $\mot{u}$. By Corollary~\ref{cor-id-follow}, there exists a
$\jungle$-word $\mot{w}$ inducing the identity, such that $\mot{sw}x$
is a $\jungle$-word. But $\mot{v}x$ and $\mot{sw}x$ induce the same
action ;  the result follows. Hence we obtain a $\jungle$-word
$\mot{u}^{(1)}$ inducing the same action as~$\mot{u}$.

 By the very same process, we can construct, for any $i \in \N$, a $\jungle$-word ${\mot{u}^{(i)}}$  inducing the same action as $\mot{u}$, such that $\mot{u}^{(1)}\mot{u}^{(2)}\hdots \mot{u}^{(i)} $ is a $ \jungle$-word. Since the set $Q^{n}$ is finite there exist $i<j, \: j-i \leq |Q|^n$, such that $\mot{u}^{(i)}$ and $\mot{u}^{(j)}$ have the same prefix of length $n$. Take $\mot{v}=\mot{u}^{(i)}\mot{u}^{(i+1)}\hdots \mot{u}^{(j-1)}$: $\mot{v}$ is a cyclic $\jungle$-word  and induces the  same action as $\mot{u}^{j-i}$. By Proposition~\ref{prop-const-order}, the order of $\rho_{\mot{v}}$ is bounded by a constant depending only on $\jungle$, hence so does $\rho_{\mot{u}}$ (with a different constant, but still depending only on $\jungle$).
Consequently, every element of $\gauta$ has a finite order, uniformly bounded by a constant, whence, as $\gauta$ is residually finite, by Zelmanov's theorem \cite{Ze1,Ze2}, $\gauta$ is finite, which concludes the proof.

\end{proof}

The tools and techniques we have developed here enabled to
  bridge the gap between~\(3\) and the set of all prime numbers. The
  next step is the extension of our result to any connected
  automaton. However, experiments suggest that there are strong similarities
  between the non prime case and the non connected case,  bringing the
hope to solve entirely the question of the (im)possible generation of
an infinite Burnside group by a reversible Mealy automaton. Note that
the primality of the stateset is not used here before
Theorem~\ref{thm:Qprefix}. It is likely that the extension of
Theorem~\ref{thm:main} to more general statesets will require to
choose carefully some $k$-self-liftable branches, with $k>1$. In fact,
there exist examples of automata for which the set of first letters in
a $\sim$-class is not the whole stateset. However the $\sim$-classes
seem to still play a crucial role in these examples. So our
construction will certainly be a key element for a more general result.

\bibliography{burnside}

\def\cprime{$'$}
\begin{thebibliography}{10}

\bibitem{AKLMP12}
A.~Akhavi, I.~Klimann, S.~Lombardy, J.~Mairesse, and M.~Picantin.
\newblock On the finiteness problem for automaton (semi)groups.
\newblock {\em Int. J. Algebr. Comput.}, 22(6):26p., 2012.

\bibitem{aleshin}
S.V. Ale{\v{s}}in.
\newblock Finite automata and the {B}urnside problem for periodic groups.
\newblock {\em Mat. Zametki}, 11:319--328, 1972.

\bibitem{dangeli_r:geometric}
D.~D'Angeli and E.~Rodaro.
\newblock A geometric approach to (semi)-groups defined by automata via dual
  transducers.
\newblock {\em Geometriae Dedicata}, 174-1:375--400, 2015.

\bibitem{gawron_ns:conjugation}
P.W. Gawron, V.V. Nekrashevych, and V.I. Sushchansky.
\newblock Conjugation in tree automorphism groups.
\newblock {\em Internat. J. Algebr. Comput.}, 11-5:529--547, 2001.

\bibitem{glushkov:automata}
V.M. Glu{\v{s}}kov.
\newblock Abstract theory of automata.
\newblock {\em Uspehi Mat. Nauk}, 16-5:3--62, 1961.

\bibitem{GKP}
Th. Godin, I.~Klimann, and M.~Picantin.
\newblock On torsion-free semigroups generated by invertible reversible {M}ealy
  automata.
\newblock In {\em LATA'15}, volume 8977 of {\em LNCS}, pages 328--339, 2015.

\bibitem{golod}
E.S. Golod.
\newblock On nil-algebras and finitely residual groups.
\newblock {\em Izv. Akad. Nauk SSSR. Ser. Mat.}, 28:273--276, 1964.

\bibitem{golod_shafarevich}
E.S. Golod and I.~Shafarevich.
\newblock On the class field tower.
\newblock {\em Izv. Akad. Nauk SSSR Ser. Mat.}, 28:261--272, 1964.

\bibitem{grigorchuk1}
R.~Grigorchuk.
\newblock On {B}urnside's problem on periodic groups.
\newblock {\em Funktsional. Anal. i Prilozhen.}, 14-1:53--54, 1980.

\bibitem{grigorch_s:ergodic_decomposition}
R.~Grigorchuk and D.~Savchuk.
\newblock Ergodic decomposition of group actions on rooted trees.
\newblock {\em to appear in Proc. of Steklov Inst. of Math.}, 2015.

\bibitem{Kli13}
I.~Klimann.
\newblock Automaton semigroups: The two-state case.
\newblock {\em Theor. Comput. Syst. (special issue STACS'13)}, pages 1--17,
  2014.
\newblock \href {http://dx.doi.org/10.1007/s00224-014-9594-0}
  {\path{doi:10.1007/s00224-014-9594-0}}.

\bibitem{klimann_ps:3state}
I.~Klimann, M.~Picantin, and D.~Savchuk.
\newblock A connected 3-state reversible {M}ealy automaton cannot generate an
  infinite {B}urnside group.
\newblock In {\em DLT' 15}, volume 9168 of {\em LNCS}, pages 313--325, 2015.

\bibitem{Ze1}
E.I. Zel{\cprime}manov.
\newblock Solution of the restricted {B}urnside problem for groups of odd
  exponent.
\newblock {\em Izv. AN SSSR Math+}, 54-1:42--59, 221, 1990.

\bibitem{Ze2}
E.I. Zel{\cprime}manov.
\newblock Solution of the restricted {B}urnside problem for {$2$}-groups.
\newblock {\em Mat. Sb.}, 182-4:568--592, 1991.

\end{thebibliography}
\end{document}